%% file: main.tex
\documentclass[a4paper,11pt]{article}
\input{Preamble.tex}

\begin{document}
\hypersetup{pageanchor=false}
\title{Fixed-Parameter Tractability of Hedge Cut\thanks{The research leading to these results has received funding from the Research Council of Norway via the project BWCA (grant no. 314528) and the European Research Council (ERC) via grant LOPPRE, reference 819416.}}

\author{
Fedor V. Fomin\thanks{
Department of Informatics, University of Bergen, Norway. Emails: \texttt{fomin@ii.uib.no}, \texttt{petr.golovach@uib.no}}
\and
Petr A. Golovach\addtocounter{footnote}{-1}\footnotemark{}
\and
Tuukka Korhonen\thanks{Department of Computer Science, University of Copenhagen, Denmark. (Work done while at the University of Bergen.) \texttt{tuko@di.ku.dk}
}\and
Daniel Lokshtanov\thanks{Department of Computer Science, University of California Santa Barbara, USA. \texttt{daniello@ucsb.edu}}
\and
Saket Saurabh\thanks{The Institute of Mathematical Sciences, Chennai, India and University of Bergen, Norway. \texttt{saket@imsc.res.in}}
}
%\author{Anonymous authors}

\maketitle

\thispagestyle{empty}

\input{abstract.tex}

\newpage
\hypersetup{pageanchor=true}
\pagestyle{plain}
\pagenumbering{arabic}

\input{intro.tex}

\input{prelims.tex}

\input{algo.tex}

\input{hardness.tex}
\input{conclusion.tex}
\bibliographystyle{alpha}
\bibliography{book_kernels_fvf}
\end{document}

%% file: abstract.tex
%!TEX root = main.tex
\begin{abstract}
In the {\hcut} problem, the edges of a graph are partitioned into groups called hedges, and the question is what is the minimum number of hedges to delete to disconnect the graph. Ghaffari, Karger, and Panigrahi~[SODA~2017] showed that \hcut can be solved in quasipolynomial-time, raising the hope for a polynomial time algorithm. Jaffke, Lima, Masar{\'{\i}}k, Pilipczuk, and Souza~[SODA~2023] complemented this result by showing that assuming the Exponential Time Hypothesis (ETH), no polynomial-time algorithm exists.
In this paper, we show that \hcut  is fixed-parameter tractable parameterized by the solution size $\ell$ by providing an algorithm with running time $\binom{\OO(\log n) + \ell}{\ell} \cdot m^{\OO(1)}$, which can be upper bounded by $c^{\ell} \cdot (n+m)^{\OO(1)}$ for any constant $c>1$. 
This running time captures at the same time the fact that the problem is quasipolynomial-time solvable, and that it is fixed-parameter tractable parameterized by $\ell$.
We further generalize this algorithm to an algorithm with running time $\binom{\OO(k \log n) + \ell}{\ell} \cdot n^{\OO(k)} \cdot m^{\OO(1)}$ for \hkcut. 
\end{abstract}

%% file: intro.tex
%!TEX root = main.tex
\section{Introduction}
Let $G$ be a graph whose edge set $E(G)$ (in this paper, we assume that $G$ could have parallel edges) has been partitioned into groups called \emph{hedges}. 
In the \hcut problem, the goal is for a given integer $\ell$ to decide whether it is possible to remove at most   $\ell$ hedges to disconnect $G$.  The problem naturally appears in several applications.     In network survivability, \hcut, known as  
{\sc Global Label Min-Cut}, models the situation when a set of links may fail simultaneously \cite{CoudertDPRV07,CoudertPRV16,DBLP:conf/soda/GhaffariKP17,ZhangF16}.  In an alternate study, 
\cite{BlinBDRS14} introduced the problem under the name {\sc Minimum Cut in Colored Graphs} in their study of problems in phylogenetics. 
  \hcut generalizes the {\sc Minimum Cut} problem in graphs (where each hedge is a single edge) and hypergraphs (where each hedge is a spanning subgraph on the vertices of a hyperedge)~\cite{DBLP:conf/soda/GhaffariKP17,DBLP:conf/soda/ChandrasekaranX18}. 

%From the theoretical perspective,  the   {\hcut} problem also possesses a  captivating nature. 
From the algorithmic complexity perspective,  the  case of {\hcut} problem is quite curious. 
Whether the problem is polynomial-time solvable or NP-complete was an open question in the literature for some years. 
Ghaffari, Karger, and Panigrahi~\cite{DBLP:conf/soda/GhaffariKP17} gave a quasipolynomial-time randomized algorithm that solves the problem in time $n^{\OO(\log \ell)}$. While the existence of a quasipolynomial-time algorithm typically provides evidence that a problem is in P, the \hcut problem does not align with this pattern. Jaffke, Lima, Masar{\'{\i}}k, Pilipczuk, and Souza~\cite{DBLP:conf/soda/JaffkeLMPS23} proved the following surprising lower bounds. Let $m$ be the total number of hedges in the graph. Then the first result of Jaffke et al. states that unless the ETH fails, there is no algorithm solving \hcut in time $(mn)^{o(\log n/(\log{\log{n}})^2)}$. The second result of Jaffke et al. concerns the parameterized complexity of the problem. They showed that the problem is  \WOH parameterized by $m-\ell$, the number of hedges unused in an optimal cut. 

However, the parameterized complexity of \hcut for its most natural parameterization, the solution size $\ell$, remained open.
For a more general problem, finding a cut in a graph $G$ minimizing the value of a non-negative integer monotone submodular cost function defined on edges of $G$, Fomin, Golovach, Korhonen, Lokshtanov, and Stamoulis~\cite{DBLP:journals/talg/FominGKLS24} ruled out the existence of a fixed-parameter tractable (\FPT) algorithm parameterized by the optimum value of the cut. The lower bound of  Fomin et al. holds even when graph $G$ is planar. For a very special class of planar graphs, Fomin et al. designed an algorithm to solve the \hcut problem in time $2^{\OO(\ell \log{\ell})}n^{\OO(1)}$.\footnote{Fomin et al. state their \FPT result for  the {\sc Hedge Cycle} problem. Due to the duality of minimal cycles and cuts in planar graphs,  their result could be restated for \hcut.} 
%For general graphs, for parameter $\ell$,  the only result known prior\todo{False} to our work is the randomized XP algorithm of Blin, Bonizzoni, Dondi,   Rizzi, and   Sikora \cite{BlinBDRS14} of running time $\OO(n^{2\ell})$.
  
We prove that \hcut is \FPT parameterized by $\ell$ by providing an algorithm for the more general \hkcut problem. In this problem, the goal is to delete edges of at most $\ell$ hedges such that the resulting graph has at least $k$ connected components. 
Chandrasekaran, Xu and Yu in~\cite[Corollary 2]{DBLP:conf/soda/ChandrasekaranX18,ChandrasekaranX21} gave an algorithm of running time 
$n^{\OO(k \log \ell)}$  for this problem. 
Our main result is the following theorem.

\begin{theorem}
\label{thm:main_algo}
There is a $\binom{\OO(k \log n) + \ell}{\ell} \cdot n^{\OO(k)} \cdot m^{\OO(1)}$ time randomized algorithm with one-sided error for \hkcut, where $\ell$ is the size of an optimum hedge $k$-cut-set, $n$ is the number of vertices, and $m$ is the number of hedges.
\end{theorem}

We note that this running time can be bounded by $c^{\ell} \cdot n^{\OO(k)} \cdot m^{\OO(1)}$ for any fixed constant $c>1$. 
This follows from considering the cases of (1) when $\ell = \Omega(k \log n)$, where $\binom{\OO(k \log n) + \ell}{\ell}$ can be bounded by $c^{\ell}$ for $c > 1$ arbitrarily close to $1$ depending on the constant in the $\Omega$-notation, and (2) when $\ell = \OO(k \log n)$, where $\binom{\OO(k \log n)+\ell}{\ell}$ can be bounded by $2^{\OO(k \log n)} = n^{\OO(k)}$.

For example, this implies that \hkcut is solvable in time $1.001^\ell \cdot  n^{\OO(k)} \cdot m^{\OO(1)}$. Thus for any fixed value $k$, the \hkcut problem is \FPT parameterized by the size $\ell$ of an optimum hedge $k$-cut-set. We refer to the book of Cygan et al. \cite{cygan2015parameterized} for an introduction to parameterized algorithms and complexity. 
Note that $\binom{\OO(k \log n) + \ell}{\ell}=\binom{\OO(k \log n) + \ell}{\OO(k \log n)}$ is also upper bounded by $n^{\OO(k \log \ell)}$, so this algorithm also runs in quasipolynomial-time for every fixed $k$, generalizing the results of~\cite{DBLP:conf/soda/GhaffariKP17,DBLP:conf/soda/ChandrasekaranX18}.

Our algorithm yields also a corresponding combinatorial upper bound on the number of optimum hedge $k$-cut-sets.

\begin{theorem}
\label{thm:main_bound}
The number of optimum hedge $k$-cut-sets of size $\ell$ is at most $\binom{\OO(k \log n) + \ell}{\ell} \cdot n^{\OO(k)}$.
\end{theorem}

Earlier, Chandrasekaran, Xu and Yu in~\cite[Theorem 4]{DBLP:conf/soda/ChandrasekaranX18,ChandrasekaranX21} showed that the number of distinct minimum hedge $k$-cut-sets of size $\ell$ is upper bounded by $n^{\OO(k \log \ell)}$. 

Finally, we show that the $n^{\OO(k)}$ factor in our algorithm is unavoidable even when the problem is parameterized by the size $\ell$ of an optimum hedge $k$-cut-set.
In particular, we show that the problem is \WOH when parameterized by $k+\ell$.
This hardness result applies even in the setting when every hedge has one component, i.e., the case of \hykcut problem.

\begin{restatable}{theorem}{hardnesstheorem}\label{thm:lower}
The \hykcut problem parameterized by both $k$ and the cut size $\ell$ is {\rm \WOH}.
\end{restatable}

%The reduction used in the proof is based on the  ideas of Marx~\cite{DBLP:journals/tcs/Marx06} for the \WO-hardness of the vertex variant of the $k$-way cut problem. To get our approximability lower bound, we combine our reduction 
% with the recent computational lower bounds of Lin et al.~\cite{DBLP:conf/icalp/LinRSW22} for the $k$-clique problem. This allows us to show that, up to \ETH, there is no algorithm with running time $f(k,\ell,\varepsilon)\cdot n^{\OO(1)}$ that for any $\varepsilon>0$, either outputs a hedge $k$-cut of size at most $(1+\varepsilon)\ell$ or correctly concludes that there is no hedge $k$-cut of size at most $\ell$. 

\paragraph*{Techniques.}
Our algorithm follows the classical idea of randomized contractions for cut problems that was pioneered by Karger and Stein~\cite{DBLP:journals/jacm/KargerS96}, extended to \hcut by Ghaffari, Karger, and Panigrahi~\cite{DBLP:conf/soda/GhaffariKP17}, and to \hkcut by Chandrasekaran, Xu, and Yu~\cite{DBLP:conf/soda/ChandrasekaranX18}.
Roughly speaking, the randomized contractions approach for \hkcut is known to work with probability at least $n^{-\OO(k)}$ in the case when the algorithm never encounters \emph{large hedges}, i.e., hedges of size at least $c \frac{n}{k}$ for some constant $c \in (0,1)$~\cite{DBLP:conf/soda/GhaffariKP17,DBLP:conf/soda/ChandrasekaranX18}.

Therefore, the challenge is to find ways to handle a large hedge, when at some step of  the contraction algorithm we encounter a large hedge.  For instance, this could happen because of the contraction of hedges in earlier steps of the algorithm. (Which can make the number of vertices of the graph smaller while maintaining the sizes of some hedges, thus making some hedges large.) Here, our approach deviates from the approaches of~\cite{DBLP:conf/soda/GhaffariKP17,DBLP:conf/soda/ChandrasekaranX18}.
Both papers~\cite{DBLP:conf/soda/GhaffariKP17,DBLP:conf/soda/ChandrasekaranX18} are based on observing that if less than $1-\varepsilon$ fraction of the large hedges are in the cut-set, then contracting a random large hedge is successful with probability~$\varepsilon$.
% and (2) if more than $1-\varepsilon$ fraction of the large hedges are cut, then we can afford to cut them all.
This will inherently lead to an $\varepsilon^{\OO(k \log n)}$ success probability, which does not result in an \FPT algorithm.

Instead, we will select one large hedge, and randomly select what happens to this hedge: either it goes to the cut, in which case the parameter $\ell$ decreases, or it is contracted, in which case the number of vertices decreases by at least $c \frac{n}{k}$.
To select one of these options, we flip a biased coin with a carefully selected bias. Then it is possible to show that any fixed optimum cut $C$ of size $|C|=\ell$ will be returned with probability  $n^{-\OO(k)}/\binom{\OO(k \log n) + \ell}{\ell}$.
While obtaining an \FPT dependence on $\ell$ like this is quite natural, we find it a bit surprising that the same algorithm recovers also the quasipolynomial algorithms and combinatorial bounds of~\cite{DBLP:conf/soda/GhaffariKP17,DBLP:conf/soda/ChandrasekaranX18}.

The reduction in the proof of \Cref{thm:lower} builds on the approach of Marx~\cite{DBLP:journals/tcs/Marx06} developed for establishing the \WO-hardness of the vertex variant of the {\sc $k$-Way Cut} problem.
%To prove the theorem, we combine our reduction 
% with the recent computational lower bounds of Lin et al.~\cite{DBLP:conf/focs/LinRSW23} for the {\sc $k$-Clique} problem. This allows us to show that assuming \ETH,
% unless {\rm FPT}$=${\rm W[1]},
 %up to \ETH, 
% there is no algorithm with running time $f(k,\ell,\varepsilon)\cdot n^{\OO(1)}$ that for any $\varepsilon>0$, either outputs a hedge $k$-cut of size at most $(1+\varepsilon)\ell$ or correctly concludes that there is no hedge $k$-cut of size at most $\ell$. 

\paragraph*{Related work.} The  {\sc Minimum Cut}, also called  {\sc Global Minimum Cut}, is one of the most well-studied problem in graph algorithms. It led to the discovery of contraction based randomized algorithms for the problem, also called Karger's algorithm~\cite{Karger93}, and the recursive contraction algorithm of Karger and Stein~\cite{DBLP:journals/jacm/KargerS96}. Queyranne~\cite{DBLP:journals/mp/Queyranne98} gave an algorithm for minimizing symmetric submodular functions. Since hypergraph cut function is symmetric and submodular, see e.g. 
\cite{Rizzi00}, this implies a polynomial time algorithm for {\sc Hypergraph Cut}. In the last few years, there have been a lot of exciting developments in improving the running times of polynomial time algorithms for cuts in hypergraphs \cite{DBLP:conf/soda/GhaffariKP17,DBLP:conf/soda/ChekuriX17,DBLP:journals/talg/FoxPZ23}.
% and  graphs with hedges. 
%these results have been extended to  hypergraphs and graphs with hedges. 
%Indeed,  Ghaffari, Karger, and Panigrahi~\cite{DBLP:conf/soda/GhaffariKP17} studied the \hcut and {\sc Hypergraph Cut} problems and obtained a quasipolynomial-time randomized algorithm, and a randomized polynomial time algorithm, respectively. Independently, Chekuri and Xu~\cite{DBLP:conf/soda/ChekuriX17} gave a faster polynomial time algorithm for  {\sc Hypergraph Cut}. 
%Fox, Panigrahi and  Zhang obtained similar results via branching randomized contraction~\cite{DBLP:journals/talg/FoxPZ23}.  
 The current fastest known polynomial time algorithm for  {\sc Hypergraph Cut} is given by 
Chekuri and Quanrud~\cite{DBLP:conf/icalp/ChekuriQ21a}.  However, the hedge cut function is not submodular and hence the generic approach based on submodular function minimization is not applicable for \hcut.

A generalization of   {\sc Minimum Cut} to {\sc $k$-Way Cut}, where the objective is to delete as few edges as possible to get at least $k$-connected components is one of the well studied problem in the literature.  Goldschmidt and Hochbaum~\cite{DBLP:journals/mor/GoldschmidtH94} showed that  the problem is \NPC, when $k$ is part of the input, and also gave a deterministic $n^{\OO(k^2)}$ time algorithm for {\sc $k$-Way Cut}. The latter result implies that {\sc  $k$-Way Cut} is polynomial-time solvable for every fixed $k$. Karger and Stein~\cite{DBLP:journals/jacm/KargerS96} designed  a faster randomized algorithm running in time $\tilde{\OO}(n^{2(k-1)})$\footnote{$\tilde{\OO}$~ hides the polylogarithimic factor in the running time.}. A series of deterministic algorithms culminates with 
Thorup's ingenious $\tilde{\OO}(n^{2k})$  time algorithm~\cite{DBLP:conf/stoc/Thorup08}. 
After a lull for some years, the problem has seen a lot of developments in recent years with respect to different algorithm paradigms meant for coping with \NPH problems.  The current best algorithm for the problem is given by He and Li~\cite{DBLP:conf/stoc/He022}
and runs in time  $\OO(n^{(1-\epsilon)k})$, where $\epsilon >0$ is an absolute constant.   Kawarabayashi and  Thorup \cite{KawarabayashiT11}  showed that  {\sc  $k$-Way Cut} is \FPT parameterized by the size of the cut edges, while the problem is known to be \WOH parameterized by $k$~\cite{DBLP:journals/entcs/DowneyEFPR03}. 

%\todo[inline]{FF: work in progress}

Chandrasekaran, Xu and Yu introduced the notion of  \hkcut in~\cite{DBLP:conf/soda/ChandrasekaranX18,ChandrasekaranX21} as the generalization of the  {\sc Hypergraph $k$-Cut}, another well-studied problem in the literature~\cite{DBLP:conf/focs/ChandrasekaranC20,BeidemanCW22}. They extended the result of Ghaffari, Karger, and Panigrahi~\cite{DBLP:conf/soda/GhaffariKP17} by showing that the number of minimum hedge $k$-cut-sets is $n^{\OO(k \log \ell)}$.

%Chandrasekaran, Xu and Yu~\cite{DBLP:conf/soda/ChandrasekaranX18} generalized the result of ~\cite{DBLP:conf/soda/GhaffariKP17} and obtained a randomized algorithm for {\sc Hypergraph $k$-Cut}. Finally, Chandrasekaran and Chekuri~\cite{DBLP:conf/focs/ChandrasekaranC20} gave a deterministic polynomial-time algorithm for {\sc Hypergraph $k$-Cut}, for fixed $k$. 
%
%Karger~\cite{DBLP:conf/soda/GhaffariKP17}, $k$-cut version~\cite{DBLP:conf/soda/ChandrasekaranX18}, %hardness~\cite{DBLP:conf/soda/JaffkeLMPS23}, submodular cycles~\cite{DBLP:conf/soda/FominGKLS23}.

Finally, we mention some known results around \hcut (apart from those we discussed earlier~\cite{DBLP:conf/soda/GhaffariKP17,DBLP:conf/soda/JaffkeLMPS23,DBLP:journals/talg/FominGKLS24}). Blin, Bonizzoni, Dondi, Rizzi, and Sikora \cite{BlinBDRS14} gave a polynomial time algorithm producing an optimum hedge min cut with probability at least $n^{-2\ell}$. That is, they gave a randomized XP algorithm for \hcut parameterized by solution size $\ell$.  The quasipolynomial algorithm of Ghaffari, Karger, and Panigrahi~\cite{DBLP:conf/soda/GhaffariKP17} 
strongly improves over this running time. 
%Coudert, Datta, P{\'{e}}rennes, Rivano, and Voge~\cite{CoudertDPRV07} showed that \hcut can be solved in time $n^{\OO(t)}$, where $t$ is an upper bound on the number of edges in any color class. 
Coudert,   P{\'{e}}rennes,   Rivano, and
 Voge \cite{CoudertPRV16}  obtained an algorithm of running time $2^c \cdot n^{\OO(1)}$, where 
 $c$ is the number of hedges with a span larger than one (hedges whose edges induce more
than one connected component).
As it was proved by Fellows, Guo, and Kanj~\cite{DBLP:journals/jcss/FellowsGK10}, the related problem $(s,t)$-\hcut (computing a minimum number of hedges required to separate vertices $s$ and $t$) is  \WTH parameterized by the number of hedges in the  cut and is \WOH 
  parameterized by the number of edges in the  cut (see also~\cite{DBLP:journals/mst/MorawietzGKS22}).

%There are several results in the literature establishing fixed-parameter tractability of \hcut with 

%When the number of edges of each color is bounded by a given constant $t$, the \hcut is solvable in time $n^{\OO(t)}$

%In yet related work, Zhang and Fu \cite{ZhangF16} parameterization \hcut by frequency and length. 
%The most relevant for us is the result of 

%\textbf{k-cut}
%$k$-cut  Kawarabayashi and  Thorup \cite{KawarabayashiT11} fpt parameterized by the size of the cut. Karger and Stein \cite{DBLP:journals/jacm/KargerS96}   ingenious contraction algorithm, a randomized algorithm with running time $\tilde{\OO}(n^{2k-1})$
%\textbf{The work on $k$-hypergraph cut. }

%% file: prelims.tex
%!TEX root = main.tex
\section{Preliminaries}
We use $\log$ to denote the natural logarithm and $[n]$ to denote the set $\{1,2,\ldots, n\}$.

Let us formally define hedge graphs.
A \emph{hedge graph} $G$ consists of a set of vertices $V(G)$ and a set of hedges $E(G)$.
Each hedge $e \in E(G)$ is a collection of disjoint subsets of $V(G)$, which are called the \emph{components} of the hedge. See \Cref{fig:example}.
We require each hedge component to be of size at least $2$ (components of size $1$ are irrelevant for the hedge cut problem).
We say that a vertex $v$ is incident with a hedge $e$ if $v$ is in some component of $e$.
The degree of a vertex is the number of hedges it is incident to.
We define the size $|e|$ of a hedge $e$ as the total number of vertices that are incident with it, i.e., the sum of the sizes of its components.

A \emph{hedge $k$-cut} of a hedge graph $G$ is a partition of $V(G)$ into $k$ non-empty sets.
A hedge $e$ is cut by a hedge $k$-cut if $e$ has a component that intersects at least two parts of the partition.
The hedge $k$-cut-set associated with a hedge $k$-cut is the set of hedges cut by the hedge $k$-cut.
A set of hedges $C \subseteq E(G)$ is a hedge $k$-cut-set if it is a hedge $k$-cut-set associated with some hedge $k$-cut.
An \emph{optimum} hedge $k$-cut-set of $G$ is a smallest hedge $k$-cut-set of $G$. In particular, if every hedge has exactly one component, then hedge $k$-cut is a $k$-cut of the hypergraph, whose hyperedges are the hedges of $G$. 

  \begin{figure}[ht]
 \center{\includegraphics[scale=0.8]{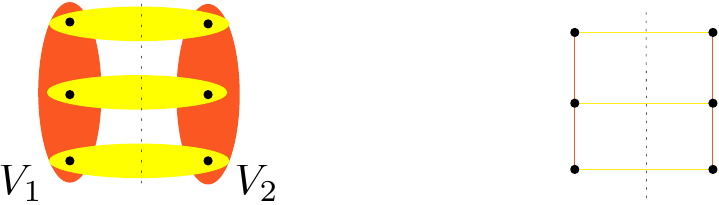}}
 \caption{Example of a hedge graph on 6 vertices. It has two hedges, yellow and red. Both hedges are of size 6. The red hedge consists of two components, each containing 3 elements. The yellow hedge has 3 components, each of size 2. The hedge 2-cut-set associated with  a 2-cut $(V_1, V_2)$ consists of all components of the yellow hedge.   }\label{fig:example}
 \end{figure}
 
These definitions of hedge graphs and hedge $k$-cut are equivalent to the definitions given in the introduction by setting for each hedge the components of that hedge be the connected components of the graph spanned by the edges in this hedge, and ignoring components of size $1$, i.e., vertices not incident to any edge of that hedge. See the graph in the right half of \Cref{fig:example}.

The contraction of a hedge $e$ means that for each component of the hedge, all occurrences of vertices of that component are identified into one vertex in all of the hedges of $G$.
This may cause some components of other hedges not to be disjoint anymore, in which case these components are merged together.
Then, all components of size one in hedges are deleted, and empty hedges, including, at this point, the contracted hedge $e$, are deleted.
Observe that because all components of a hedge have at least $2$ vertices, the contraction of $e$ decreases the number of vertices by at least $|e|/2$.
In the viewpoint where the hedges are sets of edges, the contraction of a hedge corresponds to contracting all edges of the hedge.

We say that a hedge graph $G'$ is a contraction of a hedge graph $G$ if $G'$ is obtained from $G$ by a series of contractions.
Observe that, like with contractions in graphs, the order of the contractions does not matter.
We also observe that the following properties of contractions and cuts carry over from graphs to hedge graphs.

\begin{observation}
\label{lem:contrmainten}
If $G'$ is a hedge graph that is a contraction of a hedge graph $G$, then any hedge $k$-cut-set of $G'$ is also a hedge $k$-cut-set of $G$.
\end{observation}

\begin{observation}
\label{lem:optmaint2}
Let $C$ be a hedge $k$-cut-set of $G$, and $G'$ be obtained from $G$ by contracting a hedge $e \notin C$.
Then $C$ is also a hedge $k$-cut-set of $G'$.
\end{observation}

%% file: algo.tex
%!TEX root = main.tex
\section{Algorithm}\label{sec:algo}
In this section, we give our algorithm that implies both \Cref{thm:main_algo,thm:main_bound}.
In particular, this section is devoted to the proof of the following theorem.

\begin{theorem}
\label{thm:main_aux_algo}
There is a polynomial-time randomized algorithm, that given integers $\ell$ and $k$ and a hedge graph with $n$ vertices that has an optimum hedge $k$-cut-set $C$ of size $|C| = \ell$, outputs $C$ with probability at least $n^{-\OO(k)} /\binom{\OO(k \log n) + \ell}{\ell}$.
\end{theorem}

\Cref{thm:main_aux_algo} gives the algorithm of \Cref{thm:main_algo} by running it $\binom{\OO(k \log n) + \ell}{\ell} \cdot n^{\OO(k)}$ times with increasing values of $\ell$, until a solution is found.
It implies the bound of \Cref{thm:main_bound} because the probabilities of different optimum hedge $k$-cut-sets $C$ must sum up to at most $1$.

Let us remark that the polynomial-time in \Cref{thm:main_aux_algo} means polynomial in $n+m$, where $n$ is the number of vertices and $m$ is the number of hedges.
In particular, we remark that because the graph can have parallel edges, $m$ could be unbounded as a function of $n$.

We will first describe the algorithm of \Cref{thm:main_aux_algo} in \Cref{subsec:algo:desc} and then prove its correctness in \Cref{subsec:algo:analysis}.

\subsection{Description of the algorithm}
\label{subsec:algo:desc}
The algorithm is given a hedge graph $G$ with $n = |V(G)| \ge 1$ vertices, and integers $\ell \ge 0$ and $k \ge 2$.
It returns either \fail or a hedge $k$-cut-set consisting of $\ell$ hedges.

A high-level algorithm description is given in pseudocode \Cref{alg:main}.
We note that the algorithm is a bit unintuitive, because it can return \fail if it reaches the conclusion that the size of an optimum hedge $k$-cut-set is smaller than $\ell$, particularly at \Cref{alg:main:unintuitive}.
This unintuitive behavior is needed for proving our combinatorial bound in \Cref{thm:main_bound}.

\begin{algorithm}[t]
\caption{Algorithm for \hkcut\label{alg:main}}
\textbf{Input:} Hedge graph $G$ with $n = |V(G)|$ vertices, integers $\ell$ and $k$.\\
\textbf{Output:} Either a hedge $k$-cut-set of size exactly $\ell$, or \fail.
\begin{algorithmic}[1]
\If {$n < k$} \Return \fail
\EndIf
\If {$n \le 4k$ or $\ell \le k$} \Return Brute-force($G$, $\ell$, $k$) (\Cref{lem:bruteforce})
\EndIf
\If {Exists $k-1$ vertices with degree $\le \ell/k$} \Return \fail\label{alg:main:unintuitive}
\EndIf
\If {Exists a hedge $e$ of size $|e| \ge n/(4k)$}
\State Choose one of the following with probabilities $\frac{\ell}{\lfloor 8k \cdot \log n \rfloor + \ell}$ and $\frac{\lfloor 8k \cdot \log n \rfloor}{\lfloor 8k \cdot \log n \rfloor + \ell}$, respectively:
\begin{enumerate}
\item Remove $e$ and call the algorithm recursively with $\ell$ decremented by $1$
\item Contract $e$ and call the algorithm recursively
\end{enumerate}
\Else
\State Let $G'$ be the hedge graph obtained by contracting each hedge with probability $k/(2\ell)$
\If {$|V(G')| \le n-n/80$}
\State Call the algorithm recursively on $G'$
\Else
\State \Return \fail
\EndIf
\EndIf
\end{algorithmic}
\end{algorithm}

Next, we give a detailed description of the algorithm.
The algorithm starts by handling some degenerate cases.
First, if $n < k$, then no hedge $k$-cut exists and we return \fail.
Then, if $n \le 4k$ or $\ell \le k$, a ``brute-force'' approach without recursion is used.
This case is handled by the following lemma.

\begin{lemma}
\label{lem:bruteforce}
There is a polynomial-time algorithm, that given a hedge graph with $n$ vertices and two integers $k$ and $\ell$ satisfying either $n \le 4k$ or $\ell \le k$, outputs any optimum hedge $k$-cut-set of size $\ell$ with probability at least $n^{-4k}$.
\end{lemma}
\begin{proof}
If $n \le 4k$, we assign each vertex a number from $[k]$ uniformly at random, and if the partition of the vertices resulting from this assignment gives a hedge $k$-cut with cut-set of size $\ell$, we return this cut-set.
This returns any optimum hedge $k$-cut-set of size $\ell$ with probability at least $k^{-n}$ which is at least $n^{-4k}$ because $k \le n \le 4k$.

If $\ell \le k$, we choose $\ell$ pairs of vertices $(a_1, b_1), \ldots, (a_\ell, b_\ell)$ with $a_i \neq b_i$ uniformly at random (with repetitions possible), and include to the cut-set all hedges that contain a component that is a superset of $\{a_i, b_i\}$ for some $i \in [\ell]$.
If this gives a hedge $k$-cut-set of size $\ell$, then we return it.
Because every optimum hedge $k$-cut-set of size $\ell$ can be characterized by a selection of such $\ell$ pairs of vertices, this returns any optimum hedge $k$-cut-set of size $\ell$ with probability at least $(n^2)^{-\ell} = n^{-2\ell} \ge n^{-2k}$. 
\end{proof}

Then, if there are $k-1$ vertices with degree at most $\ell/k$, we have that cutting all hedges incident with  these vertices gives a hedge $k$-cut-set of size at most $\frac{(k-1) \ell}{k} < \ell$.
In this case, we know that an optimum hedge $k$-cut-set has size less than $\ell$, and therefore we return \fail (note that the algorithm is allowed to return \fail whenever an optimum hedge $k$-cut-set has size less than $\ell$).
We observe that in the remaining cases we have that at least $n-(k-1) \ge \frac{3n}{4}$ vertices have degrees at least $\ell/k$.

We say that a hedge $e$ is \emph{large} if $|e| \ge n/(4k)$ and \emph{small} otherwise.
The algorithm works differently depending on whether there is at least one large hedge, or if all hedges are small.

If there exists a large hedge, we select arbitrarily a large hedge $e$ and probabilistically ``branch'' on $e$ in the following sense.
We choose from the options $\{\cut, \contract\}$ the option \cut with probability $\frac{\ell}{\lfloor 8k \cdot \log n \rfloor + \ell}$ and the option \contract otherwise (i.e., with probability $\frac{\lfloor 8k \cdot \log n \rfloor}{\lfloor 8k \cdot \log n \rfloor + \ell}$).
If we chose \cut, we apply the algorithm recursively on the instance with hedge $e$ removed and $\ell$ decreased by one, and if the recursive call returns a cut-set $C$, we return the union $C \cup \{e\}$, and if it returns \fail we also return \fail.
If we chose \contract, we apply the algorithm recursively on the instance with hedge $e$ contracted, and return whatever the recursive call returns.

Then, if there is no large hedge, we contract every hedge with probability $k/(2\ell)$ independently from each other to construct a contracted hedge graph $G'$.
If the contracted graph $G'$ has at most $n-n/80$ vertices, we recursively call the algorithm on $G'$ and return whatever the recursive call returns.
If $G'$ has more than $n-n/80$ vertices, we return \fail.

It is easy to see that the algorithm works in polynomial time.
It follows from \Cref{lem:contrmainten} that the algorithm always returns either a hedge $k$-cut-set of size $\ell$, or \fail.
It remains to lower bound the probability of success.

\subsection{Analysis of the algorithm}
\label{subsec:algo:analysis}
Let us now fix an optimum hedge $k$-cut-set $C$ of size $|C|=\ell$.
We will show by induction that the probability that the algorithm returns $C$ is at least $f(\ell, k, n) \cdot g(k, n)$, where the functions $f$ and $g$ are as follows:

\[f(\ell, k, n) = 1/\binom{\lfloor 8k \cdot \log n \rfloor + \ell}{\ell} \qquad \text{ and } \qquad g(k, n) = n^{-400k}.\]

Intuitively, the function $f(\ell, k, n)$ will handle the probability coming from the large hedge branching, and the function $g(k, n)$ the probability from the small hedge contraction.
The constants in these functions are chosen to optimize the simplicity of the proof instead of the running time; we believe that they could be significantly optimized with a more careful proof.
Observe that the functions $f$ and $g$ are non-increasing in all of their parameters, in particular, $f(\ell+x, k+y, n+z) \le f(\ell, k, n)$ and $g(k+x, n+y) \le g(k, n)$ hold for any non-negative $x,y,z$.

We then prove that the probability of returning $C$ is at least $f(\ell, k, n) \cdot g(k, n)$ by induction.
We first consider the base cases when the algorithm returns immediately without recursion.
Then, we prove the correctness probability by induction on $\ell+k+n$.

\paragraph*{Base cases.}
When $n < k$, no such $C$ exists.
When $n \le 4k$ or $\ell \le k$, the algorithm was described in \Cref{lem:bruteforce} and at the same time analysed to return $C$ with probability at least $n^{-4k} \ge g(k,n) \ge f(\ell,k,n) \cdot g(k,n)$.
When there are $k-1$ vertices with degree $\le \ell/k$, the optimum hedge $k$-cut-set has size less than $\ell$, and therefore again, no such $C$ exists.

\paragraph*{Large hedges.}
Now we prove the correctness probability by induction in the case when there exists a large hedge, i.e., a hedge $e$ of size at least $|e| \ge n/(4k)$.
Let $e$ be a large hedge arbitrarily chosen by the algorithm.
First, we consider the case when $e \in C$.
In this case, $C \setminus \{e\}$ is an optimum hedge $k$-cut-set of the hedge graph with the hedge $e$ removed.
Therefore, in this case, the algorithm returns $C$ if it chooses to cut $e$, and the recursive call returns $C \setminus \{e\}$.
By induction, this happens with a probability at least

\begin{align*}
&\frac{\ell}{\lfloor 8k \cdot \log n \rfloor + \ell} \cdot f(\ell-1, k, n) \cdot g(k, n)\\
&= g(k, n) /\left(\frac{\lfloor 8k \cdot \log n \rfloor + \ell}{\ell} \cdot \binom{\lfloor 8k \cdot \log n \rfloor + \ell-1}{\ell-1}\right)\\
&= g(k, n) /\binom{\lfloor 8k \cdot \log n \rfloor + \ell}{\ell} && {\text{{by} } \frac{n}{k} \cdot \binom{n-1}{k-1} = \binom{n}{k}}\\
&= f(\ell, k, n) \cdot g(k, n).
\end{align*}

So the induction assumption holds in that case.

The other case is that $e \notin C$.
In this case, by \Cref{lem:contrmainten,lem:optmaint2}, we have that $C$ is an optimum hedge $k$-cut-set of the graph obtained by contracting $e$.
Therefore, the algorithm is correct if it contracts $e$ and the recursive call returns $C$.
Note that because $|e| \ge n/(4k)$, the number of vertices of the contracted graph is at most $n-n/(8k)$.
Therefore, the overall correctness probability is by induction at least

\begin{align*}
&\frac{\lfloor 8k \cdot \log n \rfloor}{\lfloor 8k \cdot \log n \rfloor + \ell} \cdot f(\ell, k, n-n/(8k)) \cdot g(k, n-n/(8k))\\
&\ge g(k,n) /\left(\frac{\lfloor 8k \cdot \log n \rfloor + \ell}{\lfloor 8k \cdot \log n \rfloor} \cdot \binom{\lfloor 8k \cdot (\log n + \log (1-1/(8k))) \rfloor + \ell}{\ell} \right)\\
&\ge g(k,n) /\left(\frac{\lfloor 8k \cdot \log n \rfloor + \ell}{\lfloor 8k \cdot \log n \rfloor} \cdot \binom{\lfloor 8k \cdot (\log n - 1/(8k)) \rfloor + \ell}{\ell} \right) \quad \text{ by } \log (1-1/(8k)) \le -1/(8k)\\
&= g(k,n) /\left(\frac{\lfloor 8k \cdot \log n \rfloor + \ell}{\lfloor 8k \cdot \log n \rfloor} \cdot \binom{\lfloor 8k \cdot \log n \rfloor-1 + \ell}{\ell} \right)\\
&=g(k,n) /\binom{\lfloor 8k \cdot \log n \rfloor + \ell}{\ell} \qquad \text{ by } \frac{n}{n-k} \cdot \binom{n-1}{k} = \binom{n}{k}\\
&=f(\ell, k, n) \cdot g(k,n).
\end{align*}
So the induction assumption holds also in this case, which concludes the proof that it holds in the case when a large hedge exists.

\paragraph*{Small hedges.}
Now we bound the probability in the case when there are no large hedges, i.e., in the case when every hedge $e$ has $|e| < n/(4k)$.

First, we observe that with probability at least $e^{-k}$, no hedge in $C$ is contracted.

\begin{lemma}
Let $C$ be a set of hedges with $|C| = \ell$ and $\ell > k$.
Contracting every hedge with probability $k/(2\ell)$ does not contract any hedge in $C$ with probability at least $e^{-k}$.
\end{lemma}
\begin{proof}
First, because $\ell > k$, we have that the probability that any given hedge is not contracted is $1-k/(2\ell) \ge e^{-k/\ell}$.
Therefore, the probability that none of the $\ell$ hedges in the solution are contracted is at least $e^{-k}$.
\end{proof}

Then, we will show that subject to the assumption that no hedge in $C$ is contracted, there is a high probability that the number of vertices decreases by at least $n/80$.

We say that a vertex is \emph{light} if it is incident with at most $\ell/(2k)$ hedges in $C$.
Because the sum of sizes of hedges in $C$ is at most $\frac{\ell \cdot n}{4k}$, at least $n/2$ vertices are light.
We say that a vertex is \emph{good} if it is light and its degree is at least $\ell/k$.
Because at least $\frac{3n}{4}$ vertices have degree at least $\ell/k$, at least $n/4$ vertices are good.

\begin{lemma}
If a vertex is good, then the probability that a hedge incident with this vertex is contracted conditional to no hedge in $C$ being contracted is at least $1/5$.
\end{lemma}
\begin{proof}
Let $u$ be a good vertex.
First, note that even when conditioning to no hedge in $C$ being contracted, the contractions of all of the hedges not in $C$ happen with probability $k/(2\ell)$ and independently of each other.
Therefore, because $u$ is incident with at least $\ell/(2k)$ hedges not in $C$, the probability that no hedge incident with $u$ is contracted is at most
\[\left(1-\frac{k}{2\ell}\right)^{\ell/(2k)} \le e^{-\frac{k}{2 \ell} \frac{\ell}{2k}} \le e^{-1/4} \le 4/5.\]
Therefore, the probability that at least one hedge incident with $u$ is contracted is at least $1/5$.
\end{proof}

Because the number of good vertices is at least $n/4$, the expected number of vertices that are incident with a contracted hedge is at least $n/4 \cdot 1/5 \ge n/20$.
Therefore, because the maximum number of vertices incident with a contracted edge is $n$, we have that with probability at least $1/40$, at least $n/40$ vertices are incident with a contracted hedge.
In this case, the number of vertices decreases by at least $n/80$, so the algorithm makes a recursive call on the contracted instance.

We have shown that with probability at least $e^{-k}/40 \ge e^{-k-4}$, the contraction step is correct in that it does not contract any hedge in $C$, and it decreases the number of vertices by at least $n/80$.
Therefore, the overall probability that $C$ is returned is by induction at least
\begin{align*}
&e^{-k-4} \cdot f(\ell,k,n-n/80) \cdot g(k,n-n/80)\\
&\ge f(\ell,k,n) \cdot e^{-k-4} \cdot (n-n/80)^{-400k}\\
&\ge f(\ell,k,n) \cdot e^{-k-4-400k(\log n + \log(1-1/80))}\\
&\ge f(\ell,k,n) \cdot e^{-k-4-400k(\log n - 1/80)}\\
&\ge f(\ell,k,n) \cdot e^{-k-4-400k \log n + 5k}\\
&\ge f(\ell,k,n) \cdot e^{-400k \log n} \ge f(\ell,k,n) \cdot n^{-400k}.
\end{align*}

So the induction assumption holds also in this case, which concludes the proof of \Cref{thm:main_aux_algo}.

%% file: hardness.tex
%!TEX root = main.tex
\section{Hardness}
\label{sec:hardness}
In this section, we show that the \hykcut problem is \WOH when parameterized by $k$ and the cut size $\ell$. The proof follows the idea of Marx~\cite{DBLP:journals/tcs/Marx06} for the \WO-hardness of the vertex variant of the {\sc $k$-Way Cut} problem. 
%Further, we combine our reduction with the recent computational lower bounds of Lin et al.~\cite{DBLP:conf/focs/LinRSW23} for the {\sc $k$-Clique} problem and prove that the \hkcut problem does not admit an \FPT approximation scheme parameterized by $k$ and $\ell$ unless {\rm FPT}$=${\rm W[1]}. In other words, we show that, up to a basic complexity assumption, there is no algorithm with running time $f(k,\ell,\varepsilon)\cdot n^{\OO(1)}$ that for any $\varepsilon>0$, either outputs a hedge $k$-cut of size at most $(1+\varepsilon)\ell$ or correctly concludes that there is no hedge $k$-cut of size at most $\ell$.  We use the following result of~Lin et al~\cite{DBLP:conf/focs/LinRSW23}.

%Lin et al.~\cite{DBLP:conf/icalp/LinRSW22}.

%\begin{proposition}[{\cite[Theorem~1.2]{DBLP:conf/focs/LinRSW23}}]\label{prop:clique}
%There exists a function $f(k) \in k^{o(1)}$, so that assuming the \ETH, there is no \FPT algorithm parameterized by $k$, that concludes either that an input graph has no clique of size $k$ or that it has a clique of size $f(k)$.
%\end{proposition}

%Note that \Cref{prop:clique} holds for every function $f(k)$ that grows asymptotically faster than $k^{o(1)}$, as such an algorithm would imply the algorithm for the particular $f(k) \in k^{o(1)}$ in the statement.
%In particular, we will use the fact that assuming \ETH, there is no \FPT algorithm parameterized by $k$, that concludes either that an input graph has no clique of size $k$ or that it has a clique of size at least $\sqrt{k}$.

\hardnesstheorem*
\begin{proof}
As we have introduced notation only for the context of hedge graphs, let us work in the setting of hedge graphs where every hedge has one component, which is equivalent to the setting of hypergraphs.

Following Marx~\cite{DBLP:journals/tcs/Marx06}, we reduce from the {\sc $k$-Clique} problem. Let $G$ be an $n$-vertex graph and $k\geq 1$ be an integer. We assume without loss of generality that $n\geq 2k+1$ (otherwise, we add isolated vertices). We construct the hedge graph $G'$ as follows.
\begin{itemize}
\item For each edge $e\in E(G)$, construct a vertex $v_e$, and define $V(G')=V(G)\cup\{v_e\mid e\in E(G)\}$.
\item For each $x\in V(G)$, define a hedge $N_x=\{\{x\}\cup\{v_{xy}\mid y\in N_G(x)\}\}$.
\item For each pair $\{x,y\}$ of distinct $x,y\in V(G)$, make $\{\{x,y\}\}$ a hedge.
\end{itemize}
Then we set $\ell=k$ and $h=\binom{k}{2}+1$. 

We claim that $G$ has a clique of size $k$ if and only if $G'$ has a hedge $h$-cut of size at most $\ell$. 

Assume that $G$ has a $k$-clique $\{x_1,\ldots, x_k\}$. Consider the $h$-cut of $G'$ formed by the sets $V_{i,j}=\{v_{x_ix_j}\}$ for all $i,j\in[k]$ such that $i<j$ and $V^*=V(G')\setminus \big(\bigcup_{1\leq i<j\leq k}V_{ij}\big)$. 
Then $C=\{N_{x_1},\ldots, N_{x_k}\}$ is the hedge $h$-cut-set for this hedge $h$-cut and $|C|=k=\ell$. 

For the opposite direction, suppose that $G'$ has an $h$-cut $(V_1,\ldots,V_h)$ such that the corresponding $h$-cut-set $C$ is of size at most $\ell$.  First, we observe that there is $i\in[h]$ such that $V(G)\subseteq V_i$. To see this, assume, for the sake of contradiction, that some $V_i$ contains $1\leq p<n$ vertices of $G$. Then $C$ contains at least $p(n-p)\geq n-1\geq 2k>k$ hedges from the set of hedges formed by the pairs of vertices of $G$ contradicting the assumption that $|C|\leq \ell=k$. We assume without loss of generality that $V(G)\subseteq V_h$. For each $i\in[h-1]$, we select an arbitrary $u_i\in V_i$. Because each $u_i\notin V(G)$, we have that $u_i=v_{e_i}$ for some edge $e_i\in E(G)$. Denote by $X$ the set of endpoints of the edges $e_1,\ldots,e_{h-1}$ in $G$. Notice that for each $x\in X$, $N_x\in C$. Thus, $|X|\leq k$. We have that $h-1=\binom{k}{2}$ distinct edges of $G$ have their endpoints in the set $X$ of size at most $k$. This implies that $X$ is a clique of size $k$.  
\end{proof}

%% file: conclusion.tex
%!TEX root = main.tex
\section{Conclusion}\label{section:conclusion}
In this paper we gave an \FPT algorithm for the \hcut problem.
We conclude with some additional remarks and open problems.

It is natural to ask whether \hcut has a polynomial kernel.
We conjecture that the problem does not admit a polynomial kernel because it \emph{OR-composes} (see~\cite{DBLP:journals/jcss/BodlaenderDFH09,fomin2019kernelization}) to itself in the following sense.
Given a sequence of multiple instances of \hcut, each with the same parameter $\ell$, we can in polynomial-time produce a hedge graph that admits a hedge cut-set of size $\ell$ if and only if at least one of the input instances admits a hedge cut-set of size $\ell$.
This instance can be produced by simply gluing the input instances together in one vertex.
If the hedge cut problem were \NP-hard, then this composition would imply that the problem would not have polynomial kernel parameterized by $\ell$ under standard complexity assumptions for kernelization lower bounds~\cite{DBLP:journals/jcss/BodlaenderDFH09}.
This provides evidence that the problem does not admit a polynomial kernel and that perhaps this evidence could be formalized by adapting the framework for kernelization lower bounds to weaker complexity assumptions.

As for the running time, we have no evidence that the running time dependence on $\ell$ in our \FPT algorithm for \hcut is optimal.
We do not know any lower bound ruling out the existence of an algorithm of running time  $\ell^{\OO(\log\ell)} (n+m)^{\OO(1)}$.
While we obtain a $c^{\ell} (n+m)^{\OO(1)}$ running time for any fixed $c > 1$, the question whether this running time could be improved  to $2^{o(\ell)} (n+m)^{\OO(1)}$ remains open. Another interesting open question is whether our algorithm could be derandomized.

For simple graphs, that is, graphs without parallel edges, {\sc $k$-Way Cut} is fixed-parameter tractable parameterized by the size of the cut-set $\ell$ by Thorup and Kawarabayashi~\cite{KawarabayashiT11}.
This implies that for simple planar graphs, the problem is \FPT parameterized by $k$.
(Just because $\ell\leq 6k$ due to degeneracy.)
We ask whether \hkcut is \FPT on simple planar graphs parameterized by $k$?

%% file: main.bbl
\newcommand{\etalchar}[1]{$^{#1}$}
\begin{thebibliography}{MGKS22}

\bibitem[BBD{\etalchar{+}}14]{BlinBDRS14}
Guillaume Blin, Paola Bonizzoni, Riccardo Dondi, Romeo Rizzi, and Florian
  Sikora.
\newblock Complexity insights of the minimum duplication problem.
\newblock {\em Theor. Comput. Sci.}, 530:66--79, 2014.

\bibitem[BCW22]{BeidemanCW22}
Calvin Beideman, Karthekeyan Chandrasekaran, and Weihang Wang.
\newblock Deterministic enumeration of all minimum $k$-cut-sets in hypergraphs
  for fixed $k$.
\newblock In {\em Proceedings of the 2022 {ACM-SIAM} Symposium on Discrete
  Algorithms (SODA)}, pages 2208--2228, 2022.

\bibitem[BDFH09]{DBLP:journals/jcss/BodlaenderDFH09}
Hans~L. Bodlaender, Rodney~G. Downey, Michael~R. Fellows, and Danny Hermelin.
\newblock On problems without polynomial kernels.
\newblock {\em J. Comput. Syst. Sci.}, 75(8):423--434, 2009.

\bibitem[CC20]{DBLP:conf/focs/ChandrasekaranC20}
Karthekeyan Chandrasekaran and Chandra Chekuri.
\newblock Hypergraph $k$-cut for fixed $k$ in deterministic polynomial time.
\newblock In {\em Proceedings of the 61st {IEEE} Annual Symposium on
  Foundations of Computer Science (FOCS)}, pages 810--821. {IEEE}, 2020.

\bibitem[CDP{\etalchar{+}}07]{CoudertDPRV07}
David Coudert, Pallab Datta, Stephane Perennes, Herv{\'{e}} Rivano, and
  Marie{-}Emilie Voge.
\newblock Shared risk resource group complexity and approximability issues.
\newblock {\em Parallel Process. Lett.}, 17(2):169--184, 2007.

\bibitem[CFK{\etalchar{+}}15]{cygan2015parameterized}
Marek Cygan, Fedor~V. Fomin, Lukasz Kowalik, Daniel Lokshtanov, D{\'a}niel
  Marx, Marcin Pilipczuk, Micha{\l} Pilipczuk, and Saket Saurabh.
\newblock {\em Parameterized Algorithms}.
\newblock Springer, 2015.

\bibitem[CPRV16]{CoudertPRV16}
David Coudert, St{\'{e}}phane P{\'{e}}rennes, Herv{\'{e}} Rivano, and
  Marie{-}Emilie Voge.
\newblock Combinatorial optimization in networks with shared risk link groups.
\newblock {\em Discret. Math. Theor. Comput. Sci.}, 18(3), 2016.

\bibitem[CQ21]{DBLP:conf/icalp/ChekuriQ21a}
Chandra Chekuri and Kent Quanrud.
\newblock Isolating cuts, (bi-)submodularity, and faster algorithms for
  connectivity.
\newblock In {\em Proceedings of the 48th International Colloquium on Automata,
  Languages, and Programming (ICALP)}, volume 198 of {\em LIPIcs}, pages
  50:1--50:20. Schloss Dagstuhl - Leibniz-Zentrum f{\"{u}}r Informatik, 2021.

\bibitem[CX17]{DBLP:conf/soda/ChekuriX17}
Chandra Chekuri and Chao Xu.
\newblock Computing minimum cuts in hypergraphs.
\newblock In {\em Proceedings of the 28th Annual {ACM-SIAM} Symposium on
  Discrete Algorithms (SODA)}, pages 1085--1100. {SIAM}, 2017.

\bibitem[CXY18]{DBLP:conf/soda/ChandrasekaranX18}
Karthekeyan Chandrasekaran, Chao Xu, and Xilin Yu.
\newblock Hypergraph \emph{k}-cut in randomized polynomial time.
\newblock In {\em Proceedings of the 29th Annual {ACM-SIAM} Symposium on
  Discrete Algorithms (SODA)}, pages 1426--1438. {SIAM}, 2018.

\bibitem[CXY21]{ChandrasekaranX21}
Karthekeyan Chandrasekaran, Chao Xu, and Xilin Yu.
\newblock Hypergraph $k$-cut in randomized polynomial time.
\newblock {\em Math. Program.}, 186(1):85--113, 2021.

\bibitem[DEF{\etalchar{+}}03]{DBLP:journals/entcs/DowneyEFPR03}
Rodney~G. Downey, Vladimir Estivill{-}Castro, Michael~R. Fellows, Elena
  Prieto{-}Rodriguez, and Frances~A. Rosamond.
\newblock Cutting up is hard to do: the parameterized complexity of k-cut and
  related problems.
\newblock In {\em Computing: the Australasian Theory Symposiumm (CATS)},
  volume~78 of {\em Electronic Notes in Theoretical Computer Science}, pages
  209--222. Elsevier, 2003.

\bibitem[FGK10]{DBLP:journals/jcss/FellowsGK10}
Michael~R. Fellows, Jiong Guo, and Iyad~A. Kanj.
\newblock The parameterized complexity of some minimum label problems.
\newblock {\em J. Comput. Syst. Sci.}, 76(8):727--740, 2010.

\bibitem[FGK{\etalchar{+}}24]{DBLP:journals/talg/FominGKLS24}
Fedor~V. Fomin, Petr~A. Golovach, Tuukka Korhonen, Daniel Lokshtanov, and
  Giannos Stamoulis.
\newblock Shortest cycles with monotone submodular costs.
\newblock {\em {ACM} Trans. Algorithms}, 20(1):2:1--2:16, 2024.

\bibitem[FLSZ19]{fomin2019kernelization}
Fedor~V Fomin, Daniel Lokshtanov, Saket Saurabh, and Meirav Zehavi.
\newblock {\em Kernelization: {Theory} of parameterized preprocessing}.
\newblock Cambridge University Press, 2019.

\bibitem[FPZ23]{DBLP:journals/talg/FoxPZ23}
Kyle Fox, Debmalya Panigrahi, and Fred Zhang.
\newblock Minimum cut and minimum \emph{k}-cut in hypergraphs via branching
  contractions.
\newblock {\em {ACM} Trans. Algorithms}, 19(2):13:1--13:22, 2023.

\bibitem[GH94]{DBLP:journals/mor/GoldschmidtH94}
Olivier Goldschmidt and Dorit~S. Hochbaum.
\newblock A polynomial algorithm for the k-cut problem for fixed k.
\newblock {\em Math. Oper. Res.}, 19(1):24--37, 1994.

\bibitem[GKP17]{DBLP:conf/soda/GhaffariKP17}
Mohsen Ghaffari, David~R. Karger, and Debmalya Panigrahi.
\newblock Random contractions and sampling for hypergraph and hedge
  connectivity.
\newblock In {\em Proceedings of the 28th Annual {ACM-SIAM} Symposium on
  Discrete Algorithms (SODA)}, pages 1101--1114. {SIAM}, 2017.

\bibitem[HL22]{DBLP:conf/stoc/He022}
Zhiyang He and Jason Li.
\newblock Breaking the {$n^k$} barrier for minimum \emph{k}-cut on simple
  graphs.
\newblock In {\em Proceedings of the 54th Annual {ACM} {SIGACT} Symposium on
  Theory of Computing (STOC)}, pages 131--136. {ACM}, 2022.

\bibitem[JLM{\etalchar{+}}23]{DBLP:conf/soda/JaffkeLMPS23}
Lars Jaffke, Paloma~T. Lima, Tom{\'{a}}s Masar{\'{\i}}k, Marcin Pilipczuk, and
  U{\'{e}}verton~S. Souza.
\newblock A tight quasi-polynomial bound for global label min-cut.
\newblock In {\em Proceedings of the 2023 {ACM-SIAM} Symposium on Discrete
  Algorithms (SODA)}, pages 290--303. {SIAM}, 2023.

\bibitem[Kar93]{Karger93}
David~R. Karger.
\newblock Global min-cuts in rnc, and other ramifications of a simple min-cut
  algorithm.
\newblock In {\em Proceedings of the Fourth Annual {ACM/SIGACT-SIAM} Symposium
  on Discrete Algorithms (SODA)}, pages 21--30. {ACM/SIAM}, 1993.

\bibitem[KS96]{DBLP:journals/jacm/KargerS96}
David~R. Karger and Clifford Stein.
\newblock A new approach to the minimum cut problem.
\newblock {\em J. {ACM}}, 43(4):601--640, 1996.

\bibitem[KT11]{KawarabayashiT11}
Ken{-}ichi Kawarabayashi and Mikkel Thorup.
\newblock The minimum k-way cut of bounded size is fixed-parameter tractable.
\newblock In {\em Proceedings of the {IEEE} 52nd Annual Symposium on
  Foundations of Computer Science (FOCS)}, pages 160--169. {IEEE} Computer
  Society, 2011.

\bibitem[Mar06]{DBLP:journals/tcs/Marx06}
D{\'a}niel Marx.
\newblock Parameterized graph separation problems.
\newblock {\em Theoretical Computer Science}, 351(3):394--406, 2006.

\bibitem[MGKS22]{DBLP:journals/mst/MorawietzGKS22}
Nils Morawietz, Niels Gr{\"{u}}ttemeier, Christian Komusiewicz, and Frank
  Sommer.
\newblock Refined parameterizations for computing colored cuts in edge-colored
  graphs.
\newblock {\em Theory Comput. Syst.}, 66(5):1019--1045, 2022.

\bibitem[Que98]{DBLP:journals/mp/Queyranne98}
Maurice Queyranne.
\newblock Minimizing symmetric submodular functions.
\newblock {\em Math. Program.}, 82:3--12, 1998.

\bibitem[Riz00]{Rizzi00}
Romeo Rizzi.
\newblock {NOTE} - on minimizing symmetric set functions.
\newblock {\em Comb.}, 20(3):445--450, 2000.

\bibitem[Tho08]{DBLP:conf/stoc/Thorup08}
Mikkel Thorup.
\newblock Minimum k-way cuts via deterministic greedy tree packing.
\newblock In {\em Proceedings of the 40th Annual {ACM} Symposium on Theory of
  Computing (STOC)}, pages 159--166. {ACM}, 2008.

\bibitem[ZF16]{ZhangF16}
Peng Zhang and Bin Fu.
\newblock The label cut problem with respect to path length and label
  frequency.
\newblock {\em Theor. Comput. Sci.}, 648:72--83, 2016.

\end{thebibliography}
